\documentclass[12pt,oneside]{article}
\usepackage{amsfonts}
\usepackage{amsmath}
\usepackage{amssymb}
\usepackage{amsthm}
\usepackage{enumerate}
\usepackage[latin1]{inputenc} 
\usepackage{color}
\usepackage{mathtools}
\usepackage{pifont}
\usepackage{hyperref}

\newtheorem{theorem}{Theorem}[section]

\newtheorem{corollary}[theorem]{Corollary}
\newtheorem{lemma}[theorem]{Lemma}

\theoremstyle{definition}
\newtheorem{example}[theorem]{Example}

\newtheorem{remark}[theorem]{Remark}

\newenvironment{smat}{\left[\begin{smallmatrix}}{\end{smallmatrix}\right]}

\usepackage[paperwidth=210mm,paperheight=297mm,
textwidth=190mm,textheight=265mm,top=10mm,lmargin=15mm
]{geometry}

\title{The real nonnegative inverse eigenvalue problem is NP-hard 
 \footnote{Supported by the Spanish Ministerio de Ciencia y Tecnología MTM2015-68805-REDT.}}
\author{Alberto Borobia, Roberto Canogar\\
\small Dpto. Matem\'{a}ticas, Universidad Nacional de Educaci\'on a Distancia (UNED), 28040 Madrid, Spain\\
\small e-mail: $aborobia@mat.uned.es$, $rcanogar@mat.uned.es$ \\ \\
\small \url{http://dx.doi.org/10.1016/j.laa.2017.02.010}}
\date{}

\begin{document}

\maketitle

\begin{abstract}

A  list of complex numbers is realizable if it is the spectrum of a nonnegative matrix. In 1949 Sule\v{\i}manova posed the nonnegative inverse eigenvalue problem (NIEP): the problem of determining which lists of complex numbers are realizable.  The  version for reals of the NIEP (RNIEP) asks for realizable lists of real numbers. In the literature there are many sufficient conditions or criteria for lists of real numbers to be realizable. We will present an unified account of these criteria. Then we will see that the decision problem associated to the RNIEP is NP-hard and we will study the complexity for the decision problems associated to known criteria.

\end{abstract}

\section{Introduction}

A   matrix  is \emph{nonnegative} if  all its entries  are nonnegative  numbers. The \emph{Real  Nonnegative Inverse Eigenvalue Problem}  (which we will denote as {\bf RNIEP}) asks for the characterization of all possible real spectra of nonnegative matrices. A list $\Lambda=(\lambda_1,\ldots,\lambda_n)$ of $n$ real numbers is said to be \emph{realizable} if there exists some nonnegative matrix $A\geq 0$ of order $n$ with spectrum   $\sigma(A)=\{\lambda_1,\ldots,\lambda_n\}$. With some abuse of notation, 
from now on we will use the expression $\sigma(A)=\Lambda$ or $\sigma(A)=(\lambda_1,\ldots,\lambda_n)$.  


For    $\Lambda=(\lambda_1,\ldots,\lambda_n)\in \mathbb{R}^n$  define 
$$\rho(\Lambda)=\max\{|\lambda_1|, \ldots, |\lambda_n|\} \quad \text{and} \quad \Sigma(\Lambda)=\lambda_1+ \cdots + \lambda_n.$$  
We will restrict to lists of monotonically nonincreasing real numbers, that is,  elements of the sets
$$
\mathbb{R}^n_{\downarrow} \equiv 
\big\{(\lambda_1,\ldots,\lambda_n)\in \mathbb{R}^n:  \lambda_1\geq \cdots \geq \lambda_n \big\}.
$$
If $\Lambda\in \mathbb{R}_{\downarrow}^{n}$ is the  spectrum of a nonnegative matrix $A$ then   $\Sigma(\Lambda)$ is the trace of $A$ (which implies that $\Sigma(\Lambda)\geq 0$) and $\rho(\Lambda)$ is  the Perron eigenvalue of $A$ (which implies that $\rho(\Lambda)=\lambda_1$). So the candidates to be a real spectrum of some nonnegative matrix  belong to the set $\Pi_\mathbb{R}=\Pi^1_\mathbb{R}\cup \Pi^2_\mathbb{R} \cup \cdots$ where
$$
\Pi^n_\mathbb{R}\equiv \big\{\Lambda=(\lambda_1,\ldots,\lambda_n)\in \mathbb{R}_{\downarrow}^n:    \Sigma(\Lambda) \geq 0; \  \rho(\Lambda)=\lambda_1 \big\}.
$$
The set of all real spectra of nonnegative matrices is $\Pi_{\text{RNIEP}}=\Pi^1_{\text{RNIEP}}\cup \Pi^2_{\text{RNIEP}} \cup \cdots$ where
$$
\Pi^n_{\text{RNIEP}}=\{\Lambda \in \Pi^n_\mathbb{R} : \exists \text{ a nonnegative matrix $A$ of order $n$ with }  \sigma({A})=\Lambda\}.
$$
The RNIEP asks for the characterization of  $\Pi_{\text{RNIEP}}$.  The complete   characterization of $\Pi^n_{\text{RNIEP}}$ is only  known for $n\leq 4$. Indeed this seems to be an intractable problem for large $n$.  Nevertheless  several subsets of $\Pi_{\text{RNIEP}}$ are known. These partial solutions are presented in the literature as \emph{criteria}, so that if  $\Lambda\in \Pi_\mathbb{R}$ satisfies the conditions that define the criterion $\mathcal{C}$ then $\Lambda\in \Pi_{\text{RNIEP}}$.  
For each criterion $\mathcal{C}$  we define the set  $\Pi_{\mathcal{C}}=\Pi^1_{\mathcal{C}}\cup \Pi^2_{\mathcal{C}} \cup \cdots$    where
$$
\Pi^n_{\mathcal{C}} \equiv\{\Lambda \in \Pi^n_\mathbb{R} : \Lambda \text{ satisfies the condition of the criterion } \mathcal{C} \} \subset \Pi^n_{\text{RNIEP}}.
$$

The RNIEP has associated the following \emph{decision problem}: the input is  a  list $\Lambda\in \Pi_\mathbb{R}$ and the output is `yes' if  $\Lambda\in \Pi_{\text{RNIEP}}$  or `no' if $\Lambda\not\in \Pi_{\text{RNIEP}}$. Similarly, each criterion $\mathcal{C}$ has associated a decision problem where  the input is  a  list $\Lambda\in \Pi_\mathbb{R}$ and the output is `yes' if $\Lambda\in \Pi_{\mathcal{C}}$   or no' if $\Lambda\not\in \Pi_{\mathcal{C}}$. The aim of this paper is to study the complexity of  these decision problems.

\section{A review of the main criteria for the RNIEP}

For each criteria $\mathcal{C}$ we will explicitly present   the set $\Pi_{\mathcal{C}}$.  We have divided this sets into four  different groups depending on the type of conditions:

\begin{description} \label{2.1}

\item[Group 1.] {\bf Sets $\Pi_{\mathcal{C}}$  whose lists are defined by a collection of linear inequalities.}

We  introduce the following notation associated to a given $\Lambda=(\lambda_1,\ldots,\lambda_n)\in \Pi_\mathbb{R}$:
\begin{itemize}
\item $p(\Lambda)$  is the number of nonnegative elements of $\Lambda$.
\item $q(\Lambda)$  is the number of negative elements of $\Lambda$.
\item $ \Psi(\Lambda)=\{i\in \{1,\ldots,\min\{p(\Lambda),q(\Lambda) \} \}: \lambda_i+\lambda_{n+1-i}< 0\}.$
\item $\displaystyle \psi_k(\Lambda)=    \sum_{i\in \Psi(\Lambda),i<k}(\lambda_i+\lambda_{n+1-i})+\lambda_{n+1-k}    $ \quad for each   $k\in \Psi(\Lambda)$.
\item $\displaystyle \psi(\Lambda)=  \sum_{i\in \Psi(\Lambda)}(\lambda_i+\lambda_{n+1-i})+\sum_{j=p(\Lambda)+1}^{q(\Lambda)} \lambda_{n+1-j}$ \quad  (the last summation appears if $q(\Lambda) > p(\Lambda)$).
\end{itemize}

And now we present, in chronological order, the sets that belong to this group:
\begin{enumerate}[(a)]

\item  The Sule\v{\i}manova criterion~\cite{Su} gives rise to the set  
$$
\Pi_{\text{Su}}\equiv \Big\{(\lambda_1,\ldots,\lambda_n) \in \Pi_\mathbb{R} :  \lambda_1\geq 0 > \lambda_2\geq \cdots \geq \lambda_n   \Big\}.
$$

\item  The Ciarlet criterion~\cite{Ci} gives rise to the set  
$$
\Pi_{\text{Ci}}\equiv\Big\{(\lambda_1,\ldots,\lambda_n) \in \Pi_\mathbb{R} :  |\lambda_i|\leq  \frac{\lambda_1}{n} \ \text{ for } i=2,\ldots,n  \Big\}.
$$

\item The Kellogg criterion~\cite{Ke} gives rise to the  set  

$$
\Pi_{\text{Ke}}\equiv \Big\{
(\lambda_1,\ldots,\lambda_n) \in \Pi_\mathbb{R} :
\text{ if } \Gamma=(\lambda_2,\ldots,\lambda_n) \text{ then } \lambda_1 \geq  -\psi(\Gamma)  \text{ and }  \lambda_1 \geq  -\psi_k(\Gamma)  \ \forall \,  k\in \Psi(\Gamma)  \Big\}.
$$

\item The Salzmann criterion~\cite{Sa} gives rise to the  set  

$$
\Pi_{\text{Sa}}\equiv \left\{
\Lambda=(\lambda_1,\ldots,\lambda_n) \in \Pi_\mathbb{R} : \ \frac{\Sigma(\Lambda)}{n}  \geq  \frac{\lambda_i+\lambda_{n-i+1}}{2}   \ \text{ for } i=2,\ldots,  \lfloor {\frac{n+1}{2}} \rfloor  \right\}.
$$

\item The Fiedler criterion~\cite{Fi} gives rise to the  set  

$$
\Pi_{\text{Fi}}\equiv \left\{
\Lambda=(\lambda_1,\ldots,\lambda_n) \in \Pi_\mathbb{R} : \   \lambda_1 + \lambda_n + \Sigma(\Lambda) \geq 
\sum_{ 2\leq i \leq n-1}\frac{|\lambda_i+\lambda_{n-i+1}|}{2}   \  \right\}.
$$

\item The Soto-1 criterion~\cite{So1} gives rise to the  set  

$$
\Pi_{\text{So$_1$}}\equiv \Big\{ \Lambda=(\lambda_1,\ldots,\lambda_n) \in \Pi_\mathbb{R} :  \lambda_1 +\lambda_n \geq   -\psi( \Lambda) \Big\}.
$$

\end{enumerate}

\item[Group 2.] {\bf Sets $\Pi_{\mathcal{C}}$ whose lists contain a partition that satisfies some conditions.}

It is necessary to introduce some notation.  If $\Lambda_i\in \mathbb{R}^{n_i}_{\downarrow}$ for $i=1,\ldots,k$ then $\Lambda_1\cup \cdots \cup \Lambda_k$ denotes the list of $\mathbb{R}^{n_1+\cdots+n_k}_{\downarrow}$ that contains all the reals of the lists $\Lambda_1, \ldots, \Lambda_k$. For example
$$(9,-1)\cup(5,3,-4)\cup(3,3,-1,-7)=(9,5,3,3,3,-1,-1,-4,-7).$$
If $\Lambda=\Lambda_1\cup \cdots \cup \Lambda_k$ then we  say that $\Lambda_1\cup \cdots \cup \Lambda_k$ is a \emph{partition} of $\Lambda$.

For any  $\Lambda= (\lambda_1,\ldots,\lambda_n) \in \Pi^n_\mathbb{R}$ we will denote by $\Lambda^+$ the sublist  that contains all nonnegative values of $\Lambda$
and we will denote by $\Lambda^-$ the sublist that contains all negative values of $\Lambda$.
Note that $\Lambda=\Lambda^+\cup \Lambda^-$.

Now we write, in chronological order,  the sets that belong to this group:

\begin{enumerate}[(a)]

\item The Sule\v{\i}manova-Perfect criterion~\cite{Su,Pe1} gives rise to the  set  
$$
\Pi_{\text{SP}} \equiv\Big\{\Lambda \in \Pi_\mathbb{R} : \exists \text{ a partition }  \Lambda=\Lambda_1\cup\cdots\cup\Lambda_k  \text{ such that } \Lambda_1,\ldots,\Lambda_k \in  \Pi_{\text{Su}}  \Big\}.
$$

 \item The Perfect-1 criterion~\cite{Pe1} gives rise to the set  
$$
\Pi_{\text{Pe$_1$}}\equiv \Big\{
\Lambda \in \Pi_\mathbb{R} : \exists \text{ a partition } \Lambda=(\alpha,\beta)\cup \Lambda_1\cup\cdots\cup\Lambda_k \text{ such that } \alpha=\rho(\Lambda),\ \beta\leq 0,  \text{ and }  
$$
$$
 \text{ for } i=1,\ldots,k \ \ 
\Lambda_i=(\lambda_i,\lambda_{i1},\ldots,\lambda_{it_i}) 
\text{ with } \Sigma(\Lambda_i)\leq 0, \   \lambda_{i1},\ldots,\lambda_{it_i}\leq 0 \leq \lambda_i, \text{ and }   \lambda_i+\beta\leq 0
\Big\}.
$$

\item The Borobia criterion~\cite{Bo} gives rise to the  set  
$$
\Pi_{\text{Bo}}\equiv \Big\{
\Lambda \in \Pi_\mathbb{R} : \exists \text{ a partition } \Lambda^-=\Lambda_1\cup\cdots\cup\Lambda_k \text{ such that } \Lambda^+\cup(\Sigma(\Lambda_1))\cup \cdots\cup (\Sigma(\Lambda_k)) \in \Pi_\text{Ke}  \Big\}.
$$

\end{enumerate}


\item[Group 3.] {\bf Sets $\Pi_{\mathcal{C}}$ whose lists are defined recursively. }

On this category appear the subsets of $\Pi_\text{RNIEP}$ associated to four different criteria whose authors are: (a) Soules~\cite{Sou} (we consider the extended Soules criterion as presented in Section 2.1 of~\cite{SE}); (b) Borobia, Moro and Soto~\cite{BMS}; (c) Soto~\cite{So}; and (d) \v{S}migoc and Ellard~\cite{Sm,SE}.  The exposition of these criteria is quite elaborate. So we will refer to~\cite{SE} where it is made a detailed presentation of each  criterion and it is proved that all the four criteria are equivalent, that is, that the four sets   associated to the criteria are equal: $\Pi_\text{Sou}=\Pi_\text{BMS}=\Pi_\text{So}=\Pi_\text{SE}$.   Therefore here we only need to expose  one of them. We have chosen the criterion given by Borobia, Moro and Soto since   it has the simplest recursive definition. Consider the following assertions:
\begin{enumerate} [(i)] 
\item If $\Lambda_1\in \Pi_{\text{RNIEP}}$ and $\Lambda_2\in \Pi_{\text{RNIEP}}$  then $\Lambda_1\cup \Lambda_2\in \Pi_{\text{RNIEP}}$.
\item If $(\lambda_1,\lambda_2,\ldots,\lambda_n)\in \Pi_{\text{RNIEP}}$  then $(\lambda_1+\epsilon,\lambda_2,\ldots,\lambda_n)\in \Pi_{\text{RNIEP}}$ for any $\epsilon>0$.
\item If $(\lambda_1,\ldots,\lambda_i,\ldots,\lambda_n)\in \Pi_{\text{RNIEP}}$  then $(\lambda_1+\epsilon,\ldots,\lambda_i \pm \epsilon,\ldots,\lambda_n)\in \Pi_{\text{RNIEP}}$ for any $\epsilon>0$.
\end{enumerate}
That (i) and (ii) are true is  well known, and the proof of  (iii) is due to  Guo~\cite{Gu}.  A list $(\lambda_1,\ldots,\lambda_n)$ is   \emph{C-realizable} if it may be obtained by starting with the $n$ trivially realizable lists (0),(0),\ldots,(0) and then using  (i), (ii) and (iii) any number of times in any order. Associated to this recursive construction is the set  
$$\Pi_{\text{BMS}} \equiv\{ \Lambda \in \Pi_\mathbb{R} : \Lambda \text{ is C-realizable} \}.$$
An interesting and open question about $\Pi_{\text{BMS}}$ is if the subsets  $\Pi^n_{\text{BMS}} =\Pi_{\text{BMS}}\cap  \Pi^n_\mathbb{R}$ could  be  characterized by a list of linear inequalities for any $n$.

\item[Group 4.] {\bf Sets $\Pi_{\mathcal{C}}$ whose lists are dependent on the existence of  a nonnegative matrix with  prescribed diagonal and  spectrum.}

Again it is necessary  to introduce some notation.  If $\Lambda_2\in \mathbb{R}^{n_2}_{\downarrow}$ is a sublist of $\Lambda_1\in \mathbb{R}^{n_1}_{\downarrow}$ then $\Lambda_1\setminus \Lambda_2\in \mathbb{R}^{n_1-n_2}_{\downarrow}$ denotes the sublist of $\Lambda_1$ such that $\Lambda_1=\Lambda_2\cup (\Lambda_1\setminus \Lambda_2)$. For example
$$(8,6,3,3,3,-4,-4,-6,-7)\setminus(8,3,-4,-4,-7)=(6,3,3,-6).$$

The Perfect-2$^+$ criterion~\cite{Pe2} gives rise to the set
{\small$$
\Pi_{\text{Pe$_2^+$}} \equiv 
\Big\{\Lambda\in \Pi_\mathbb{R} : 
\Lambda=\Big(\left[\Lambda_1\setminus \big(\rho(\Lambda_1)\big)\right]\cup (\alpha_1)\Big)\cup\cdots\cup \Big(\big[\Lambda_k\setminus \big(\rho(\Lambda_k)\big)\big]\cup (\alpha_k) \Big)
\text{ for some }  \Lambda_1,\ldots,\Lambda_k\in \Pi_{\text{Su}} $$
$$\text{and some } \alpha_1=\rho(\Lambda),\alpha_2,\ldots,\alpha_k\geq 0,   \text{ and } \exists  A=\begin{smat} \rho(\Lambda_1) & & * \\ & \ddots & \\ * & & \rho(\Lambda_k)\end{smat}
\geq 0 \text{ with } \sigma(A)=(\alpha_1,\ldots,\alpha_k) 
\Big\}. 
$$}


\end{description}

Marijuan, Pisonero and Soto~\cite{MPS}  analyzed the relationships between the sets $\Pi_{\mathcal{C}}$. Recently,  Ellard and \v{S}migoc~\cite{SE}  have completed the analysis. Figure 1 reflects all  relationships and updates the original map found in~\cite{MPS}.

\begin{figure}[ht] \vspace{-0.5cm}
\centering
\includegraphics[width=16cm]{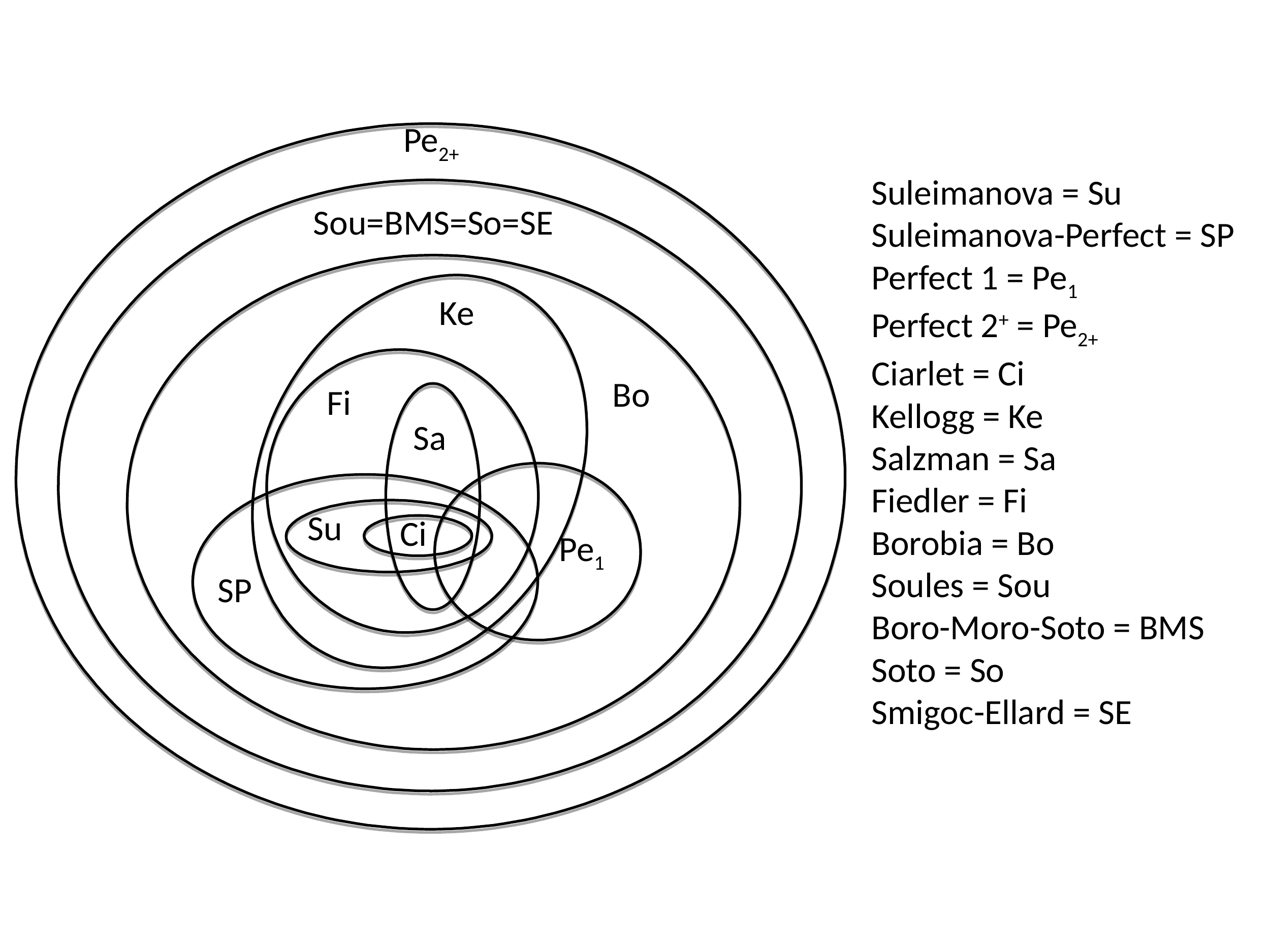}
\vspace{-15mm}
\caption[]{\small{$\Pi_\text{SP} \subset \Pi_\text{Bo} \subset \Pi_\text{Sou}=  \Pi_\text{BMS}= \Pi_\text{So}= \Pi_\text{SE} \subset \Pi_\text{Pe$_{2^+}$}  \subset \Pi_\text{RNIEP}$}}
\end{figure}

\section{The {$(\Pi_1,\Pi_2)-$Problem}}

Complexity classes (P, NP, NP-complete, NP-hard) are collections of decision problems: problems whose inputs can be answered by `yes' or `no'. Actually a decision problem is completely described by the inputs for which the answer is `yes'. To be more precise, consider two sets $\Pi_1$ and $\Pi_2$  such that $\Pi_1 \subset \Pi_2$. The \emph{$(\Pi_1,\Pi_2)-$Problem} is the decision problem in which the input is an element $\Lambda\in \Pi_2$ and the output is `yes' if $\Lambda\in \Pi_1$ or  `no' if $\Lambda \not \in \Pi_1$. The set $\Pi_1$ will be the set we are interested in, and the set $\Pi_2$ can be thought as the \emph{context} of the problem (that is, a set that imposes some minimum requirements to its elements to discard trivial non-elements of $\Pi_1$). Now we will present the decision problems that will appear in this work.

%
%

\subsubsection*{The $(\Pi_{\text{RNIEP}},\Pi_\mathbb{R})-${Problem}}

We are interested in the decision version of the RNIEP, that is, in the $(\Pi_{\text{RNIEP}},\Pi_\mathbb{R})-$\emph{Problem} for which  the input is an element $\Lambda\in \Pi_\mathbb{R}$ and the output is `yes' if $\Lambda\in \Pi_{\text{RNIEP}}$ or  `no' if $\Lambda \not \in \Pi_{\text{RNIEP}}$. 

\begin{example} 
The following are two  instances of the $(\Pi_{\text{RNIEP}},\Pi_\mathbb{R})-$Problem:
\begin{enumerate}[(i)]
\item  $(6,-3)\in \Pi_\mathbb{R}$ is a yes-instance  of the $(\Pi_{\text{RNIEP}},\Pi_\mathbb{R})-$Problem since   it is the spectrum of $\begin{smat} 1 & 5 \\ 4 & 2 \end{smat}$.
\item  $\Lambda=(3,3,-2-2,-2)\in \Pi_\mathbb{R}$  is a no-instance of the $(\Pi_{\text{RNIEP}},\Pi_\mathbb{R})-$Problem.  Suppose $A$ is a nonnegative matrix of order 5 with spectrum $\Lambda$. By the Perron-Frobenius Theorem, $A$ is reducible and $\Lambda$ can be partitioned into two nonempty lists each one being the spectrum of a nonnegative matrix with Perron eigenvalue equal to 3. This is not possible since one of the sublists must contain numbers with negative sum.  So $\Lambda$ is a no-instance.
\end{enumerate}
\end{example}


\subsubsection*{The $(\Pi_{\mathcal{C}},\Pi_\mathbb{R})-${Problem} where $\mathcal{C}$ is a criterion of realizability}

For each  criterion $\mathcal{C}$ of realizability we have constructed the set  $\Pi_{\mathcal{C}}$ where  $\Pi_{\mathcal{C}}\subset \Pi_{\text{RNIEP}}\subset \Pi_\mathbb{R}$. So, associated to $\mathcal{C}$ we have  the  $(\Pi_{\mathcal{C}},\Pi_\mathbb{R})-$\emph{Problem}, a decision problem for which  the input is an element $\Lambda\in \Pi_\mathbb{R}$ and the output is `yes' if $\Lambda\in \Pi_{\mathcal{C}}$ or  `no' if $\Lambda \not \in \Pi_{\mathcal{C}}$. 

\begin{example} 
We will consider the same instance   for two different $(\Pi_\mathcal{C},\Pi_\mathbb{R})-$Problems:
\begin{enumerate}[(i)]
\item  $(4,2,-3,-3)\in \Pi_\mathbb{R}$ is a no-instance  of the $(\Pi_{\text{SP}},\Pi_\mathbb{R})-$Problem  since it can not be partitioned in Sule\v{\i}manova sets.  
\item  $(4,2,-3,-3)\in \Pi_\mathbb{R}$ is a yes-instance  of the $(\Pi_{\text{BMS}},\Pi_\mathbb{R})-$Problem because of  the sequence:
\begin{enumerate}[1.]
\item $(0),(0),(0),(0)$.
\item $(0,0),(0),(0)$.
\item $(3,-3),(0),(0)$.
\item $(3,-3),(0,0)$.
\item $(3,-3),(2,-2)$.
\item $(3,2,-2,-3)$.
\item $(4,2,-3,-3)$.
\end{enumerate}
\end{enumerate}
\end{example}

\begin{remark} \label{restrictions}
In the $(\Pi_{\text{RNIEP}},\Pi_\mathbb{R})-${Problem} and in the $(\Pi_{\mathcal{C}},\Pi_\mathbb{R})-${Problem}  the context is $\Pi_\mathbb{R}$. To analyze the complexity of a decision problem it is important that the context has as few restrictions as possible, otherwise we are hiding part of the complexity since the instances are already preselected. In our case the restrictions are the following: $(i)$ the numbers are ordered, $(ii)$ $\Sigma (\Lambda)\geq 0$, and $(iii)$ $\rho(\Lambda)=\lambda_1$. 
 
What happens if the context is the set $\mathbb{R}\cup \mathbb{R}^2\cup \cdots$ of   lists of real numbers without further restrictions? For a list with $n$ elements the cost of ordering its elements is $n \log(n)$, the cost of checking condition $(ii)$ is $n$, and the cost of checking condition $(iii)$ is unitary since it is only necessary to check that $\lambda_1\geq |\lambda_n|$. So the overall process is $n \log(n)$. This is the hidden  part of the complexity when the context is $\Pi_\mathbb{R}$.

Having said that, the reason of considering $\Pi_\mathbb{R}$ as the context is because it makes the exposition clearer.


\end{remark}
\subsubsection*{The Partition Problem}

Let $\Pi_\mathbb{N}$ be  the  set of lists of non-increasing positive integers, that is,  
$$
\Pi_\mathbb{N}=\Pi^1_\mathbb{N}\cup \Pi^2_\mathbb{N}\cup \cdots \quad \text{where} \quad \Pi^n_\mathbb{N} \equiv \big\{(i_1,\ldots,i_n) : i_1,\ldots,i_n\in \mathbb{N}; \ i_1\geq \cdots \geq i_n>0   \big\}.
$$
And consider the set
$$
\Pi_{\text{PP}}\equiv\{I\in \Pi_\mathbb{N} : \exists \text{ a partition }  I=J\cup K \text{ such that  }  \Sigma(J)=\Sigma(K)\}.
$$
As $\Pi_{\text{PP}}\subset \Pi_\mathbb{N}$ then it makes sense to consider the $(\Pi_{\text{PP}},\Pi_\mathbb{N})-$Problem. Indeed this is a well known decision problem that in the literature is known as the \emph{Partition Problem}\footnote{For an interesting and nontechnical presentation of the Partition Problem see~\cite{Ha}.}. The input of the Partition Problem is usually a list of unordered positive integers, but the restriction to ordered list does not change the complexity of the Partition Problem. 

In what follows we will use $(\Pi_{\text{PP}},\Pi_\mathbb{N})-$Problem or Partition Problem  interchangeably.

\begin{example}  The following are two  instances of the Partition Problem:

\begin{enumerate}[(i)]
\item $(9,6,4,4,2,1)\in \Pi_\mathbb{N}$ is a yes-instance of the $(\Pi_{\text{PP}},\Pi_\mathbb{N})-$Problem since  $$(9,6,4,4,2,1)=(9,4)\cup(6,4,2,1) \quad  \text{with } 9+4=6+4+2+1.$$
\item   $(8,6,4,1)\in\Pi_\mathbb{N}$ is a no-instance of  the $(\Pi_{\text{PP}},\Pi_\mathbb{N})-$Problem since  the sum of its integers is odd.
\end{enumerate}
\end{example}

\section{NP-hardness of the RNIEP}\label{CompCompl}

A decision problem  is in the class {\bf NP-hard} when every decision problem in the class {\bf NP} (nondeterministic polynomial-time) can be reduced in polynomial time to it\footnote{A good reference for computational complexity theory  is~\cite{AB}.}.
We will  prove that the $(\Pi_{\text{RNIEP}},\Pi_\mathbb{R})-$Problem  is  NP-hard. 
Actually we will see that   for any arbitrary set $X$ such that $\Pi_{\text{SP}} \subseteq X \subseteq \Pi_{\text{RNIEP}}$ the  $(X,\Pi_\mathbb{R})-$Problem is  NP-hard. This will be done by using the technique of reducing a problem that is known to be NP-hard, the Partition Problem, to our decision problem. 


 
\begin{lemma}  \label{PPisReducibletoX}
The $(\Pi_{\text{PP}},\Pi_\mathbb{N})-$Problem is reducible to the  $(X,\Pi_\mathbb{R})-$Problem for   $\Pi_{\text{SP}} \subseteq X \subseteq \Pi_{\text{RNIEP}}$. 
\end{lemma}

\begin{proof}
Define the function 
$$
\begin{matrix}
\phi: & \Pi_\mathbb{N} & \longrightarrow &  \Pi_\mathbb{R} \\
& I=(i_1,\ldots,i_n) & \mapsto & \phi(I)=(\frac{\Sigma(I)}{2},\frac{\Sigma(I)}{2},-i_n,\ldots,-i_1)
\end{matrix}
$$
It is clear that the  transformation of $I$ into $\phi(I)$ is  done in linear time with respect to the size of the input.  
It remains  to prove that $I\in \Pi_\mathbb{N}$ is a yes-instance for the $(\Pi_{\text{PP}},\Pi_\mathbb{N})-$Problem if and only if $\phi(I)\in \Pi_\mathbb{R}$ is a yes-instance for the $(X,\Pi_\mathbb{R})-$Problem. That is, to prove for $I\in \Pi_\mathbb{N}$ that  $I\in \Pi_\text{PP}$   if and only if $\phi(I)\in X$:

\begin{itemize}

\item If $I\in \Pi_\text{PP}$ then  $\phi(I) \in X$.

If $I\in \Pi_\text{PP}$ then there exist a partition $I=J\cup K=(j_{1},\ldots,j_{p})\cup (k_{1},\ldots,k_{q})$ with  $\Sigma(J)=\Sigma(K)=\Sigma(I)/2$. Thus 
$$
\begin{matrix}
\phi(I)=(\frac{\Sigma(I)}{2},-j_{p},\ldots,-j_{1})\cup (\frac{\Sigma(I)}{2},-k_{q},\ldots,-k_{1})
\end{matrix}
$$
with    $(\frac{\Sigma(I)}{2},-j_{p},\ldots,-j_{1})\in \Pi_{\text{Su}}$ and $(\frac{\Sigma(I)}{2},-k_{q},\ldots,-k_{1})\in \Pi_{\text{Su}}$. Then  $\phi(I) \in \Pi_{\text{SP}} \subseteq X$.

\item If $\phi(I)\in X$ then $I\in \Pi_\text{PP}$.

If $\phi(I)\in X$ then $\phi(I)\in \Pi_\text{RNIEP}$ since $X \subseteq \Pi_{\text{RNIEP}}$.
Then there exists a nonnegative matrix $A$ whose spectrum is $\phi(I)$.  The Perron root of an irreducible nonnegative matrix is its spectral radius and has algebraic multiplicity one. As the spectral radius of $A$ is $\Sigma(I)/2$ and  it appears twice in $\phi(I)$ then  $A$ is reducible. This implies that there exists a permutation matrix $P$ such  that 
$$P^TAP=\begin{bmatrix} A_1 & * \\ 0 & A_2 \end{bmatrix}$$
where $A_1$ and $A_2$ have spectral radius $\Sigma(I)/2$.  So 
$$
\begin{matrix}
\sigma(A)=\sigma(A_1)\cup \sigma(A_2)=(\frac{\Sigma(I)}{2},-\alpha_{1},\ldots,-\alpha_{r})\cup (\frac{\Sigma(I)}{2},-\beta_{1},\ldots,-\beta_{s})=\phi(I).
\end{matrix}
$$
Therefore  $I=(\alpha_{r},\ldots,\alpha_{1})\cup(\beta_{s},\ldots,\beta_{1})$. As $A\geq 0$  and the trace of $A$ is zero (this is because $\Sigma({\phi(I)})=0$)  then all the entries on the diagonal of $A$ are equal to zero, and so all the entries on the diagonals of $A_1$ and  $A_2$ are   equal to zero. So the traces of  $A_1$ and  $A_2$ are equal to zero and then
$$
\begin{matrix}
\frac{\Sigma(I)}{2}=\alpha_{1}+\cdots+\alpha_{r}=\beta_{1}+\cdots+\beta_{s}
\end{matrix}
$$ 
and  so $I\in \Pi_{\text{PP}}$. 

\end{itemize}
\end{proof}

When a problem that is known to be NP-hard (like the famous list of Karp's 21 NP-complete problems~\cite{Ka}, which includes the partition problem) is reducible to a new problem, then automatically, the new problem becomes NP-hard because of a standard argument that we reproduce below.

\begin{theorem} \label{XisNPHard}
The  $(X,\Pi_\mathbb{R})-$Problem is NP-hard for   $\Pi_{\text{SP}} \subseteq X \subseteq \Pi_{\text{RNIEP}}$.
\end{theorem} 

\begin{proof}
 All NP problems are reducible in polynomial time to the Partition Problem (see  Karp~\cite{Ka}). On the other hand, in Lemma~\ref{PPisReducibletoX} we have seen that the Partition Problem is reducible in polynomial time  to the $(X,\Pi_\mathbb{R})-$Problem.  Therefore, by the transitivity of the reduction relation, every problem in the class NP is reducible in polynomial time  to the $(X,\Pi_\mathbb{R})-$Problem. Thus the $(X,\Pi_\mathbb{R})-$Problem is NP-hard. 
\end{proof}

\begin{corollary} \label{CisNPHard} (i) The  $(\Pi_{\text{RNIEP}},\Pi_\mathbb{R})-$Problem is NP-hard.\\
(ii) The  $(\Pi_{\mathcal{C}},\Pi_\mathbb{R})-$Problem is NP-hard for  ${\mathcal{C}}$=SP, Bo, Sou, BMS,  So, SE and  Pe$_{2^+}$.   
\end{corollary} 

\begin{proof}
It is important to notice that  all criteria ${\mathcal{C}}$=SP, Bo, Sou, BMS,  So, SE and  Pe$_{2^+}$ contain  SP  as Figure 1 shows. The result follows from Theorem~\ref{XisNPHard}.
\end{proof}

The criteria Pe$_1$ did not fit well into the scheme of Lemma~\ref{PPisReducibletoX}, so we will treat it separately.

\begin{theorem}\label{Pe1}
The $(\Pi_{\text{Pe$_1$}},\Pi_\mathbb{R})-$Problem is NP-hard. 
\end{theorem}

We  outline the proof. First we prove that the Partition Problem is reducible to the  $(\Pi_{Pe_1},\Pi_\mathbb{R})-$Problem.  The function that gives rise to  the reduction is 
the function 
$$
\begin{matrix}
\phi: & \Pi_\mathbb{N} & \longrightarrow &  \Pi_\mathbb{R} \\
& I=(i_1,\ldots,i_n) & \mapsto & \phi(I)=(\frac{\Sigma(I)}{2},\frac{\Sigma(I)}{2},\frac{\Sigma(I)}{2},-i_n,\ldots,-i_1,-\frac{\Sigma(I)}{2})
\end{matrix}
$$
Note that 
$$
\begin{matrix}
\phi(I)=(\frac{\Sigma(I)}{2},-\frac{\Sigma(I)}{2})\cup (\frac{\Sigma(I)}{2},-j_{p},\ldots,-j_{1})\cup (\frac{\Sigma(I)}{2},-k_{q},\ldots,-k_{1})
\end{matrix}
$$
 satisfies the conditions that define $\Pi_{\text{Pe$_1$}}$.  To finish the proof  we  argue as in the proof of Theorem~\ref{XisNPHard}.

\section{{The RNIEP  and  the decision problems for rationals}}

In Section~\ref{CompCompl} we have stablished the NP-hardness of several decision problems: the ones corresponding to RNIEP and Groups 2, 3 and 4. The following natural question is to ask if we can be more specific and prove that all these decision problems are NP-complete (a decision problem is  NP-complete if it is NP-hard and NP). 
For the rest of the decision problems not treated (the ones corresponding to Group 1), we will see if they are in the class {\bf P} (polynomial-time solvable).

We require a specification. To determine if a decision problem belongs to the class   P  or to determine  if it belongs to the seemingly broader class NP  implies typically that we deal with discrete problems,  over the integers or rationals, about graphs, etc. This inadequateness of complexity theory to treat problems for real and complex numbers is well explained in~\cite{Bl}. The computational problems that arise in the RNIEP have as  domain the reals. To \emph{discretize} the $(\Pi_{\text{RNIEP}},\Pi_\mathbb{R})-$Problem  (but remain as faithful  to the original problem as possible) in what follows we will consider its rational version. Namely, if
$$
\Pi_\mathbb{Q}=\{(\lambda_1,\ldots,\lambda_n) \in \Pi_\mathbb{R} : \lambda_1,\ldots,\lambda_n\in \mathbb{Q}\} 
$$
and 
$$
\Pi_{\text{QNIEP}}= \Pi_{\text{RNIEP}} \cap \Pi_\mathbb{Q}
$$
then the $(\Pi_{\text{QNIEP}},\Pi_\mathbb{Q})-$Problem is the \emph{rational version} of the $(\Pi_{\text{RNIEP}},\Pi_\mathbb{R})-$Problem. In a similar way, for any given criterion  $\mathcal{C}$  if 
$$
\Pi_{\mathcal{C}(\mathbb{Q})}=\{(\lambda_1,\ldots,\lambda_n) \in \Pi_{\mathcal{C}} : \lambda_1,\ldots,\lambda_n\in \mathbb{Q}\}= \Pi_{\mathcal{C}} \cap \Pi_\mathbb{Q}
$$
then the $(\Pi_{\mathcal{C}(\mathbb{Q})},\Pi_\mathbb{Q})-$Problem is the \emph{rational version} of the $(\Pi_{\mathcal{C}},\Pi_\mathbb{R})-$Problem. 

If we reproduce the content of  Section~\ref{CompCompl}  considering  rationals instead of reals, then we conclude that Theorems~\ref{XisNPHard} and~\ref{Pe1} are also valid in the following rational version.

\begin{theorem} \label{XisNPHardQ}
The  $(X,\Pi_\mathbb{Q})-$Problem is NP-hard for   $\Pi_{\text{SP}(\mathbb{Q})} \subseteq X \subseteq \Pi_{\text{QNIEP}}$.
\end{theorem} 

\begin{theorem}\label{Pe1Q}
The $(\Pi_{\text{Pe$_1$}(\mathbb{Q})},\Pi_\mathbb{Q})-$Problem is NP-hard. 
\end{theorem}

Finally, we are ready to review all the decision problems and determine their complexity class:

\subsubsection*{\underline{The complexity of the QNIEP}}

We do not know how to decide if the $(\Pi_{\text{QNIEP}},\Pi_\mathbb{Q})-$Problem  belongs to the class NP. The difficulty is that a certificate of membership in NP is apparently the solution to the  problem. That is, for a yes-instance $\Lambda\in\Pi_{\text{QNIEP}}$ of the  $(\Pi_{\text{QNIEP}},\Pi_\mathbb{Q})-$Problem the certificate would be a nonnegative matrix $A\geq 0$ with $\sigma(A)=\Lambda$.  As the entries of $A$ are real numbers  then to check that $\Lambda$ is the spectrum of $A$ can not be done, in general,  in polynomial time.

\subsubsection*{\underline{The complexity of the criteria of Group 1}}

\begin{theorem}
The  $(\Pi_{\mathcal{C}(\mathbb{Q})},\Pi_\mathbb{Q})-$Problem  is in  the class P for    ${\mathcal{C}}$=Su, Ci, Ke, Sa,  Fi and So$_1$.
\end{theorem} 

\begin{proof}
If  $\Lambda=(\lambda_1,\ldots,\lambda_n)\in \Pi_\mathbb{Q}$ then $\Lambda\in \Pi_{\mathcal{C}(\mathbb{Q})}$ if and only if $\Lambda\in {\Pi}^n_{\mathcal{C}(\mathbb{Q})}$.  That is, $\Lambda\in \Pi_{\mathcal{C}(\mathbb{Q})}$ if and only if  $\lambda_1,\ldots,\lambda_n$ satisfies the linear inequalities that defines $\Pi^n_{\mathcal{C}(\mathbb{Q})}$. Observe that in all the cases ${\Pi}^n_{\mathcal{C}(\mathbb{Q})}$ is defined by a collection of at most $n$  inequalities. Therefore, the overall process to check that $\Lambda\in \Pi_{\mathcal{C}(\mathbb{Q})}$ will employ at most quadratic time with respect to the size of the input $\Lambda$. We conclude that   the $(\Pi_{\mathcal{C}(\mathbb{Q})},\Pi_\mathbb{Q})-$Problem    belongs to the class P. 
\end{proof}

\subsubsection*{\underline{The complexity of the criteria of Groups 2 and 3}}


\begin{theorem} \label{NPs}
The  $(\Pi_{\mathcal{C}(\mathbb{Q})},\Pi_\mathbb{Q})-$Problem is  NP-complete for  
${\mathcal{C}}$=SP, Bo, Sou, BMS,  So, SE and Pe$_1$.  
\end{theorem}
\begin{proof} 
All these decision problems are NP-hard by  Theorems~\ref{XisNPHardQ} and~\ref{Pe1Q}. 

Let us see that they are also NP.  A decision problem belongs to the class NP  if for each yes-instance there exists a \emph{certificate}  that  can be checked  in  polynomial time.

Let $\Lambda\in\Pi_\mathbb{Q}$ be a yes-instance for the  $( \Pi_{\text{SP}(\mathbb{Q})},\Pi_\mathbb{Q})-$Problem. Take as certificate for $\Lambda$  any partition  $\Lambda=\Lambda_1\cup\cdots\cup\Lambda_k$   such that $\Lambda_1,\ldots,\Lambda_k \in  \Pi_{\text{Su}(\mathbb{Q})}$.  Checking that $\Lambda_1,\ldots,\Lambda_k \in  \Pi_{\text{Su}(\mathbb{Q})}$  can be done in linear time. So the $(\Pi_{\text{SP}(\mathbb{Q})},\Pi_\mathbb{Q})-$Problem belongs to the class  NP.

Let $\Lambda\in\Pi_\mathbb{Q}$ be a yes-instance for the  $(\Pi_{\text{Bo}(\mathbb{Q})},\Pi_\mathbb{Q})-$Problem. Take as certificate for $\Lambda$  any partition  $\Lambda^-=\Lambda_1\cup\cdots\cup\Lambda_k$   such that  $\Lambda^+\cup(\Sigma(\Lambda_1))\cup \cdots\cup (\Sigma(\Lambda_k)) \in \Pi_{\text{Ke}(\mathbb{Q})}$. Checking  that $\Lambda^+\cup(\Sigma(\Lambda_1))\cup \cdots\cup (\Sigma(\Lambda_k)) \in \Pi_{\text{Ke}(\mathbb{Q})}$  can be done in polynomial time. So the $(\Pi_{\text{Bo}(\mathbb{Q})},\Pi_\mathbb{Q})-$Problem belongs to the class  NP.

Let $\Lambda=(\lambda_1,\ldots,\lambda_n)\in\Pi_\mathbb{Q}$ be a yes-instance for the  $(\Pi_{\text{BMS}(\mathbb{Q})},\Pi_\mathbb{Q})-$Problem. Take as certificate for $\Lambda$  any sequence of the three allowed moves that transform  the $n$ trivially realizable lists (0),(0),\ldots,(0) into the list $(\lambda_1,\ldots,\lambda_n)$. Checking that the moves  perform this transformation can be done in polynomial time. So the $(\Pi_{\text{BMS}(\mathbb{Q})},\Pi_\mathbb{Q})-$Problem belongs to the class NP.

Consider   ${\mathcal{C}}$=Sou, So or SE. As $\Pi_{\text{BMS}(\mathbb{Q})}=\Pi_{\text{Sou}(\mathbb{Q})}=\Pi_{\text{So}(\mathbb{Q})}=\Pi_{\text{SE}(\mathbb{Q})}$ then   take as certificate  for the $(\Pi_{\mathcal{C}(\mathbb{Q})},\Pi_\mathbb{Q})-$Problem  the same   certificate than for the $(\Pi_{\text{BMS}(\mathbb{Q})},\Pi_\mathbb{Q})-$Problem. So the $(\Pi_{\mathcal{C}(\mathbb{Q})},\Pi_\mathbb{Q})-$Problem belongs to the class NP.

Let $\Lambda\in\Pi_\mathbb{Q}$ be a yes-instance for the  $( \Pi_{\text{Pe$_1$}(\mathbb{Q})},\Pi_\mathbb{Q})-$Problem. Take as certificate for $\Lambda$  any partition  $\Lambda=(\alpha,\beta)\cup \Lambda_1\cup\cdots\cup\Lambda_k$    that satisfies the condition  in the definition of $\Pi_{\text{Pe$_1$}(\mathbb{Q})}$.  Checking those conditions  can be done in linear time. So the $(\Pi_{\text{Pe$_1$}(\mathbb{Q})},\Pi_\mathbb{Q})-$Problem belongs to the class  NP.
\end{proof}

\subsubsection*{\underline{The complexity of the criterion of Group 4}}

That the  $(\Pi_{\text{Pe$_{2^+}$}(\mathbb{Q})},\Pi_\mathbb{Q})-$Problem is NP-hard is an immediate  consequence of Theorem~\ref{XisNPHardQ}.  But as in the case of the $(\Pi_{\text{QNIEP}},\Pi_\mathbb{Q})-$Problem, for the $(\Pi_{\text{Pe$_{2^+}$}(\mathbb{Q})},\Pi_\mathbb{Q})-$Problem we do not know if it belongs to the class NP, since a certificate of membership in NP includes apparently a nonnegative matrix with prescribed diagonal and spectrum.

\subsubsection*{\underline{Discussion}}
The situation described in the complexity for the QNIEP and for the criterion of Group 4 lead us to ask the following question: 
\begin{quote} \em
Let $A$ be  a square nonnegative matrix  whose  eigenvalues are rational numbers. 
Does  there always exist a rational nonnegative matrix $B$ with the same spectrum than $A$?
\end{quote}


 This question resembles the one posed by Cohen and Rothblum~\cite{CR} related to the nonnegative matrix factorization of a rational nonnegative matrix. This question is restated by Vavasis~\cite{Va} as follows: ``\emph{Suppose an $m\times n$ rational matrix $A$ has nonnegative rank $k$ and a corresponding nonnegative factorization $A =W H$, $W \in \mathbb{R}^{m\times k}$, $H \in \mathbb{R}^{k\times n}$. Is it guaranteed that there exist rational $W, H$ with the same properties?}'' Interestingly,  Vavasis proves that the nonnegative matrix  factorization  is NP-hard.


Now let us consider  a similar question   focused on the coefficients of the characteristic polynomial. We will use the result of Kim, Ormes and Roush~\cite{Kim} who proved the Spectral Conjecture of Boyle and Hendelman~\cite{Boy} for rationals. Suppose that A is a positive matrix (no necessarily with rational entries) of order $n$ whose characteristic polynomial $f(\lambda) $ has rational coefficients. The Spectral Theorem for rationals implies that  there exists a nonnegative matrix with rational entries of a sufficient large order $N$  whose characteristic polynomial is $\lambda^{N-n}f(\lambda)$.  This  lead us to also ask the following:

\begin{quote} \em
Let A be a square nonnegative matrix whose characteristic polynomial has rational coefficients. 
Does there always exist a rational nonnegative matrix B with the same characteristic polynomial than A?%
\end{quote}

\paragraph{Acknowledgements:} We thank an anonymous referee who suggested to include the last question.

\bibliographystyle{plain}

\begin{thebibliography}{1}

\small{

\bibitem{AB}
S. Arora, B. Barak
\newblock Computational complexity: a modern approach. 
\newblock {\em Cambrigde University Press, 2009.}**


\bibitem{Bl}
L. Blum,
\newblock Computing over the reals: where Turing meets Newton,
\newblock {\em Notices Amer. Math. Soc}, 51 (2004), no. 9, 1024--1034.


\bibitem{Bo}
A. Borobia,
\newblock On the nonnegative eigenvalue problem,
\newblock {\em Linear Algebra Appl.}, 223/224 (1995) 131--140 (Special issue honoring Miroslav Fiedler and Vlastimil Pták).

\bibitem{BMS}
A. Borobia, J. Moro, R.L. Soto,
\newblock A unified view on compensation criteria in the real nonnegative inverse eigenvalue problem,
\newblock {\em Linear Algebra Appl.}, 428(11--12) (2008) 2574--2584.

\bibitem{Boy}
M. Boyle, D. Hendelman 
\newblock The spectra of nonnegative matrices via symbolic dynamics.
\newblock {\em Ann. of Math.} (2) 133 (1991), no. 2, 249-316.

\bibitem{Ci}
P. Ciarlet,
\newblock Some results in the theory of nonnegative matrices,
\newblock {\em Linear Algebra Appl.}, 1(1) (1968) 139--152.


\bibitem{CR}
J. Cohen and U. Rothblum,
\newblock Nonnegative ranks, decompositions and factorizations of nonnegative matrices,
\newblock {\em Linear Algebra Appl.}, 190 (1993), pp. 149--168.


\bibitem{SE}
R. Ellard, H. \v{S}migoc,
\newblock Connecting sufficient conditions for the Symmetric Nonnegative Inverse Eigenvalue Problem
\newblock {\em Linear Algebra Appl.}, 498 (2016) 521--552.


\bibitem{Fi}
M. Fiedler,
\newblock Eigenvalues of nonnegative symmetric matrices,
\newblock {\em Linear Algebra Appl.},  9 (1974) 119--142.


\bibitem{Gu}
W. Guo,
\newblock Eigenvalues of nonnegative matrices,
\newblock {\em Linear Algebra Appl.}, 266 (1997) 261--270.


\bibitem{Ha}
B. Hayes,
\newblock Computing science: The easiest hard problem,
\newblock {\em Am. Scientist}, 90 (2002), 113?117.


\bibitem{Ka}
R.M. Karp, 
\newblock Reducibility among combinatorial problems,
\newblock {\em Complexity of computer computations (Proc. Sympos., IBM Thomas J. Watson Res. Center, Yorktown Heights, N.Y., 1972), pp. 85--103. Plenum, New York, 1972.}


\bibitem{Ke}
R. Kellogg,
\newblock Matrices similar to a positive or essentially positive matrix,
\newblock {\em Linear Algebra Appl.}, 4 (1971) 191--204.

\bibitem{Kim}
K.H. Kim, N.S. Ormes, F.W. Roush 
\newblock The spectra of nonnegative integer matrices via formal power series
\newblock {\em Jour. Amer. Math. Soc.} 10 (2000) 773-806



\bibitem{MPS}
C. Marijuán, M. Pisonero, R.L. Soto,
\newblock A map of sufficient conditions for the real nonnegative inverse eigenvalue problem,
\newblock {\em Linear Algebra Appl.}, 426(2--3) (2007) 690--705.


\bibitem{Pe1}
H. Perfect,
\newblock Methods of constructing certain stochastic matrices,
\newblock {\em Duke Math. J.}, 20 (1953) 395--404.


\bibitem{Pe2}
H. Perfect,
\newblock Methods of constructing certain stochastic matrices II,
\newblock {\em Duke Math. J.}, 22 (1955) 305--311.


\bibitem{Sa}
F. Salzmann,
\newblock A note on the eigenvalues of nonnegative matrices,
\newblock {\em Linear Algebra Appl.}, 5 (1972) 329--338.


\bibitem{Sm}
H. \v{S}migoc,
\newblock The inverse eigenvalue problem for nonnegative matrices,
\newblock {\em Linear Algebra Appl.}, 393 (2004) 365--374 (Special issue on Positivity in Linear Algebra).


\bibitem{So1}
R.L. Soto,
\newblock Existence and construction of nonnegative matrices with prescribed spectrum,
\newblock {\em Linear Algebra Appl.}, 369 (2003) 169--184.


\bibitem{So}
R.L. Soto,
\newblock A family of realizability criteria for the real and symmetric nonnegative inverse eigenvalue problem,
\newblock {\em Numer. Linear Algebra Appl.}, 20(2) (2013) 336--348.


\bibitem{Sou}
G. Soules,
\newblock Constructing symmetric nonnegative matrices,
\newblock {\em Linear Multilinear Algebra}, 13 (1983) 241--251.


\bibitem{Su}
H. Sule\v{\i}manova,
\newblock Stochastic matrices with real characteristic values,
\newblock {\em Dokl. Akad. Nauk SSSR}, 66 (1949) 343--345 (in Russian).


\bibitem{Va}
S.A. Vavasis,
\newblock On that complexity of nonnegative matrix factorization,
\newblock {\em SIAM J. Optim.}, 20 (2009) no. 3, 1364--1377.

}

\end{thebibliography}

\end{document}